\documentclass[conference]{IEEEtran}
\IEEEoverridecommandlockouts
\usepackage{cite}
\usepackage{amsmath,amssymb,amsfonts}
\usepackage{algorithmic}
\usepackage{graphicx}
\usepackage{textcomp}
\usepackage{xcolor}
\usepackage{bbm}
\usepackage{amsthm}
\usepackage{psfrag,caption,subcaption}
\usepackage{dblfloatfix}    

\usepackage{bbm}
\usepackage{acronym,comment}

\acrodef{ris}[RIS]{reconfigurable intelligent surface}
\acrodef{isac}[ISAC]{Integrated sensing and communication}
\acrodef{bs}[ISAC]{base station}
\acrodef{dfbs}[DFBS]{dual function radar communication base station}
\acrodef{ue}[UE]{user equipment}
\acrodef{sinr}[SINR]{signal-to-interference-plus-noise ratio}
\acrodef{snr}[SNR]{signal-to-noise ratio}
\acrodef{mui}[MUI]{multi-user-interference}
\acrodef{ula}[ULA]{uniform linear array}
\acrodef{mimo}[MIMO]{multiple input multiple output}
\acrodef{LoS}[LoS]{line-of-sight}
\acrodef{NLoS}[NLoS]{non-line-of-sight}
\acrodef{dfrc}[DFRC]{dual function radar communication}

\usepackage[utf8]{inputenc}
\usepackage[english]{babel}
\newtheorem{theorem}{Theorem}

\def\cast{{
   \mathord{
      \hbox to 0em{
         \ooalign{
	   \smash{\hbox{$\ast$}}\crcr
	   \smash{\hskip-1pt\Large\hbox{$\circ$}} }
	 \hidewidth}
      \phantom{\bigcirc}
} }}

\def\bm#1{\mbox{\boldmath $#1$}}

\newcommand{\rH}{^{ \raisebox{1pt}{$\rm \scriptscriptstyle H$}}}
\newcommand{\rT}{^{ \raisebox{1.2pt}{$\rm \scriptstyle T$}}}

\newcommand{\bds}{\begin {itemize}}
\newcommand{\eds}{\end {itemize}}
\newcommand{\bdf}{\begin{definition}}
\newcommand{\blm}{\begin{lemma}}
\newcommand{\edf}{\end{definition}}
\newcommand{\elm}{\end{lemma}}
\newcommand{\bthm}{\begin{theorem}}
\newcommand{\ethm}{\end{theorem}}
\newcommand{\bprp}{\begin{prop}}
\newcommand{\eprp}{\end{prop}}
\newcommand{\bcl}{\begin{claim}}
\newcommand{\ecl}{\end{claim}}
\newcommand{\bcr}{\begin{coro}}
\newcommand{\ecr}{\end{coro}}
\newcommand{\bquest}{\begin{question}}
\newcommand{\equest}{\end{question}}

\newcommand{\larrow}{{\larrow}}

\newcommand{\argmin}{\ensuremath{\mathrm{arg}\min}}
\newcommand{\argmax}{\ensuremath{\mathrm{arg}\max}}

\newcommand{\cA}{{\ensuremath{\mathcal{A}}}}

\newcommand{\cC}{{\ensuremath{\mathcal{C}}}}

\newcommand{\cN}{{\ensuremath{\mathcal{N}}}}

\newcommand{\cP}{{\ensuremath{\mathcal{P}}}}

\def\mbC{{\ensuremath{\mathbb C}}}

\def\mbE{{\ensuremath{\mathbb E}}}

\newcommand{\va}{{\ensuremath{{\mathbf{a}}}}}

\newcommand{\vc}{{\ensuremath{{\mathbf{c}}}}}
\newcommand{\vd}{{\ensuremath{{\mathbf{d}}}}}

\newcommand{\vg}{{\ensuremath{{\mathbf{g}}}}}

\newcommand{\vh}{{\ensuremath{{\mathbf{h}}}}}

\newcommand{\vn}{{\ensuremath{{\mathbf{n}}}}}

\newcommand{\vs}{{\ensuremath{{\mathbf{s}}}}}

\newcommand{\vt}{{\ensuremath{{\mathbf{t}}}}}
\newcommand{\vu}{{\ensuremath{{\mathbf{u}}}}}
\newcommand{\vv}{{\ensuremath{{\mathbf{v}}}}}

\newcommand{\vx}{{\ensuremath{{\mathbf{x}}}}}

\newcommand{\mA}{{\ensuremath{\mathbf{A}}}}

\newcommand{\mB}{{\ensuremath{\mathbf{B}}}}

\newcommand{\mC}{{\ensuremath{\mathbf{C}}}}

\newcommand{\mE}{{\ensuremath{\mathbf{E}}}}

\newcommand{\mH}{{\ensuremath{\mathbf{H}}}}

\newcommand{\mI}{{\ensuremath{\mathbf{I}}}}

\newcommand{\mJ}{{\ensuremath{\mathbf{J}}}}

\newcommand{\mR}{{\ensuremath{\mathbf{R}}}}

\newcommand{\mS}{{\ensuremath{\mathbf{S}}}}

\newcommand{\mT}{{\ensuremath{\mathbf{T}}}}

\newcommand{\mV}{{\ensuremath{\mathbf{V}}}}

\usepackage{latexsym}

\def\IC{\mathbb C}
\def\IN{\mathbb N}
\def\IZ{\mathbb Z}
\def\IR{\mathbb R}

\def\shat{^{\mathchoice{}{}%
 {\,\,\smash{\hbox{\lower4pt\hbox{$\widehat{\null}$}}}}%
 {\,\smash{\hbox{\lower3pt\hbox{$\hat{\null}$}}}}}}

\def\bSigma{{
      \ooalign{
      \smash{\hskip.4pt\raise.4pt\hbox{$\Sigma$}}\vphantom{}\crcr
      \smash{\hskip.7pt\raise.6pt\hbox{$\Sigma$}}\vphantom{}\crcr
      \smash{\hbox{$\Sigma$}}\vphantom{$\Sigma$}}
      \vphantom{\hbox{$\Sigma$}}
      }}
\def\bTheta{{
      \ooalign{
      \smash{\hskip.5pt\raise.5pt\hbox{$\Theta$}}\vphantom{}\crcr
      \smash{\hskip.0pt\raise.1pt\hbox{$\Theta$}}\vphantom{}\crcr
      \smash{\hbox{$\Theta$}}\vphantom{$\Theta$}}
      \vphantom{\hbox{$\Theta$}}
      }}
\def\bDelta{{
      \ooalign{
      \smash{\hskip.4pt\raise.4pt\hbox{$\Delta$}}\vphantom{}\crcr
      \smash{\hskip.7pt\raise.6pt\hbox{$\Delta$}}\vphantom{}\crcr
      \smash{\hbox{$\Delta$}}\vphantom{$\Delta$}}
      \vphantom{\hbox{$\Delta$}}
      }}
\def\bLambda{{
      \ooalign{
      \smash{\hskip.5pt\raise.5pt\hbox{$\Lambda$}}\vphantom{}\crcr
      \smash{\hskip.0pt\raise.1pt\hbox{$\Lambda$}}\vphantom{}\crcr
      \smash{\hbox{$\Lambda$}}\vphantom{$\Lambda$}}
      \vphantom{\hbox{$\Lambda$}}
      }}

\makeatletter

\def\bordermatrix#1{\begingroup \m@th
  \@tempdima 8.75\p@
  \setbox\z@\vbox{%
    \def\cr{\crcr\noalign{\kern2\p@\global\let\cr\endline}}%
    \ialign{$##$\hfil\kern2\p@\kern\@tempdima&\thinspace\hfil$##$\hfil
      &&\quad\hfil$##$\hfil\crcr
      \omit\strut\hfil\crcr\noalign{\kern-\baselineskip}%
      #1\crcr\omit\strut\cr}}%
  \setbox\tw@\vbox{\unvcopy\z@\global\setbox\@ne\lastbox}%
  \setbox\tw@\hbox{\unhbox\@ne\unskip\global\setbox\@ne\lastbox}%
  \setbox\tw@\hbox{$\kern\wd\@ne\kern-\@tempdima\left[\kern-\wd\@ne
    \global\setbox\@ne\vbox{\box\@ne\kern2\p@}%
    \vcenter{\kern-\ht\@ne\unvbox\z@\kern-\baselineskip}\,\right]$}%
  \null\;\vbox{\kern\ht\@ne\box\tw@}\endgroup}
\makeatother

\makeatletter
\def\argmin{\mathop{\operator@font arg\,min}}
\def\argmax{\mathop{\operator@font arg\,max}}
\makeatother

\def\bm#1{\mbox{\boldmath $#1$}}

\newcommand{\bbeta}{{\bm\beta}}

\newcommand{\bea}{\begin{array}}
\newcommand{\ena}{\end{array}}
\newcommand{\beq}{\begin{equation}}
\newcommand{\enq}{\end{equation}}

\newcommand{\beqa}{\begin{eqnarray}}
\newcommand{\enqa}{\end{eqnarray}}

\newcommand{\beqan}{\begin{eqnarray*}}
\newcommand{\enqan}{\end{eqnarray*}}

\newcommand{\AL}{\begin{enumerate}}
\newcommand{\ALE}{\end{enumerate}}

\def\addots{\mathinner{
    \mkern1mu\raise0pt\vbox{\kern7pt\hbox{.}}
    \mkern2mu\raise4pt\hbox{.}
    \mkern2mu\raise7pt\hbox{.}
    \mkern1mu}}

\def\sddots{\mathinner{
    \mkern.8mu\raise7pt\hbox{.}
    \mkern.8mu\raise4pt\hbox{.}
    \mkern.8mu\raise0pt\vbox{\kern7pt\hbox{.}}
    \mkern1mu}}

\def\saddots{\mathinner{
    \mkern.2mu\raise0pt\vbox{\kern7pt\hbox{.}}
    \mkern.2mu\raise4pt\hbox{.}
    \mkern.2mu\raise7pt\hbox{.}
    \mkern1mu}}

\def\sqplus{\mathbin{
	{\ooalign{\hfil\raise.3ex\hbox{\scriptsize
	+}\hfil\crcr\mathhexbox274\crcr\mathhexbox275}}
	}} 
\def\sqminus{\mathbin{
	{\ooalign{\hfil\raise.3ex\hbox{\scriptsize
	--}\hfil\crcr\mathhexbox274\crcr\mathhexbox275}}
	}}

\def\IC{{
   \mathord{
      \hbox to 0em{
	 \hskip-4pt
         \ooalign{
	   \smash{\hskip1.9pt\raise2.6pt\hbox{$\scriptscriptstyle |$}}\crcr
	   \smash{\hbox{\rm\sf C}} }
	 \hidewidth}
      \phantom{\hbox{\rm\sf C}}
} }}
\def\IN{
    {\ooalign{
   \smash{\hskip2.2pt\raise1.5pt\hbox{$\scriptscriptstyle |$}}\vphantom{}\crcr
   \hbox{\sf N}
	}}
	} 
\def\IZ{
    {\ooalign{
   \smash{\hskip1.9pt\raise0pt\hbox{$\sf Z$}}\vphantom{}\crcr
   \hbox{\sf Z}
	}}
	} 
\def\IR{
    {\ooalign{
   \smash{\hskip2.2pt\raise1.5pt\hbox{$\scriptscriptstyle |$}}\vphantom{}\crcr
   \smash{\hskip2.2pt\raise3.3pt\hbox{$\scriptscriptstyle |$}}\vphantom{}\crcr
   \hbox{\sf R}
	}}
	} 

\DeclareMathAlphabet{\mathcmb}{OT1}{cmr}{b}{n}

\def\bSigma{\ensuremath{\mathcmb{\Sigma}}}
\def\bLambda{\ensuremath{\mathcmb{\Lambda}}}

\def\bTheta{\ensuremath{\mathcmb{\Theta}}}

\newcommand{\SI}{\begin{indlist}}
\newcommand{\EI}{\end{indlist}}

\newcommand{\DL}{\begin{dashlist}}
\newcommand{\DLE}{\end{dashlist}}

\makeatletter
\def\setboxz@h{\setbox\z@\hbox}
\def\wdz@{\wd\z@}
\def\boxz@{\box\z@}
\def\underset#1#2{\binrel@{#2}%
  \binrel@@{\mathop{\kern\z@#2}\limits_{#1}}}
\def\binrel@#1{\begingroup
  \setboxz@h{\thinmuskip0mu
    \medmuskip\m@ne mu\thickmuskip\@ne mu
    \setbox\tw@\hbox{$#1\m@th$}\kern-\wd\tw@
    ${}#1{}\m@th$}%
  \edef\@tempa{\endgroup\let\noexpand\binrel@@
    \ifdim\wdz@<\z@ \mathbin
    \else\ifdim\wdz@>\z@ \mathrel
    \else \relax\fi\fi}%
  \@tempa
}
\let\binrel@@\relax%
\makeatother

\def\BibTeX{{\rm B\kern-.05em{\sc i\kern-.025em b}\kern-.08em
		T\kern-.1667em\lower.7ex\hbox{E}\kern-.125emX}}
\begin{document}
	
	\title{
		\color{black} Beamforming in Hybrid RIS assisted Integrated Sensing and Communication Systems
	}
	
	\author{\IEEEauthorblockN{ R.S. Prasobh Sankar and  Sundeep Prabhakar Chepuri\\ }
		\IEEEauthorblockA{ {Indian Institute of Science,}
			Bangalore, India 
		}
	}
	\maketitle

	\begin{abstract}
		In this paper, we consider a hybrid reconfigurable intelligent surface~(RIS) comprising of active and passive elements to aid an integrated sensing and communication~(ISAC) system serving multiple users and targets. Active elements in a hybrid RIS include amplifiers and phase shifters, whereas passive elements include only phase shifters. We jointly design transmit beamformers and RIS coefficients, i.e., amplifier gains and phase shifts, to maximize the worst-case target illumination power while ensuring a desired signal-to-interference-plus-noise ratio for  communication links and constraining the RIS noise power due to the active elements. Since this design problem is not convex, we propose a solver based on alternating optimization  to design the transmit beamformers and RIS coefficients. Through numerical simulations, we demonstrate that the performance of the proposed hybrid RIS assisted ISAC system is significantly better than that of passive RIS assisted ISAC systems as well as ISAC systems without RIS even when only a small fraction of the hybrid RIS contains active elements.
		
	\end{abstract}
	
	\begin{IEEEkeywords}
		Active RIS,  hybrid RIS, integrated sensing and communication, reconfigurable intelligent surfaces, transmit beamforming. 
	\end{IEEEkeywords}
	
	\section{Introduction} \label{sec:intro}
	\ac{isac}  systems are envisioned to play a key role in 6G systems~\cite{liu2021ISAC_6G,wymeersch2021Integration_comm_sense_6G}. ISAC systems aim to establish reliable communication links with users while sharing the same spectral resources to simultaneously perform sensing tasks. The coexistence of communication and sensing functionalities is usually accomplished by using communication symbols for sensing, embedding communication data in radar waveforms, or by simultaneously transmitting precoded communication symbols and radar waveforms ~\cite{mishra2019towards_mmWave,liu2020joint_transmit_beamform}. Dual function radar communication is a type of an ISAC system with a common \ac{dfbs} carrying out both sensing and communication tasks. 
	
	Large bandwidth at mmWave frequencies enables higher data rates and precise positioning, thereby making the mmWave frequency band attractive for \ac{isac} systems~\cite{mishra2019towards_mmWave}. However, operating at mmWave frequencies is challenging due to the extreme pathloss, which often renders the \ac{NLoS} paths very weak.  \Acp{ris} are a promising technology, which can favorably modify the wireless propagation environment by introducing additional paths~\cite{renzo2020smart,basar2019wireless}. Typically, an \ac{ris} comprises of a fully passive array of phase shifters, which can be  remotely tuned to introduce certain phase shifts to the incident signal. Passive \acp{ris}, which are envisaged as one of the crucial technology for wireless communication systems~\cite{rajatheva2020white}, have unsurprisingly received significant research interest for sensing as well as ISAC systems~\cite{sankar2021mmWave_JRC,wang2021ris_radar, wang2021joint_waveform_design,liu2021joint,song2021JRC_RIS,he2022risassisted}. However, a common problem with passive \acp{ris} is the so-called  \textit{double fading}, where the effective pathloss of the link via the RIS is the product of pathlosses of the transmitter-RIS and RIS-receiver links. Hence, the improvement in communication or sensing performance due to the passive RIS is not significant  when the direct link itself is sufficiently strong.
	
	In contrast to passive \acp{ris}, \textit{active RISs} have phase shifters and reflection-type amplifiers that are capable of amplifying the incident signals~\cite{zhang2022active,dong2021active_ris_secure}.  Active \acp{ris} do not comprise of any radio frequency~(RF) chains and are different from systems with passive phase shifters replaced with fully active sensors or decode-and-forward relays~\cite{nguyen2021hybrid}.  Active \acp{ris} 
	offer remarkable improvement in performance when compared with traditional passive \acp{ris} for communications~\cite{zhang2022active,dong2021active_ris_secure}. However, active amplifier elements introduce additional noise, referred to as \textit{RIS noise}, into the system, and should be accounted for while designing systems with active \acp{ris}.
	
	In this work, we consider a \textit{hybrid RIS}, which comprises of both active and passive elements, but without any RF chains to decode or process impinging signals. In particular, we consider a hybrid RIS-assisted \ac{isac} system, wherein a \ac{dfbs} is communicating with multiple users while sensing multiple targets. Specifically, we design  transmit beamformers at the \ac{dfbs} and coefficients of the hybrid RIS to maximize the worst-case target illumination power at targets while ensuring a minimum~\ac{sinr} for the communication \acp{ue} and constraining the amount of RIS noise at targets. Since we consider a hybrid RIS,  modeling and design are different compared to existing works~\cite{wang2021joint_waveform_design,liu2021joint,song2021JRC_RIS,he2022risassisted}, which focus on designing passive phase shifters without RIS noise.
	
	Since this design problem is a non-convex optimization problem, we propose an alternating optimization procedure to compute the transmit beamformers and RIS coefficients. Subproblems pertaining to the design of beamformers for fixed RIS coefficients and vice versa are also non-convex. Therefore, we relax the subproblems and obtain semi-definite programs~(SDPs), which can be solved using off-the-shelf convex solvers. The impact of RIS noise is accounted for in both the subproblems. Through numerical simulations, we demonstrate that the performance of the proposed ISAC system with hybrid RIS is significantly better than that of passive RIS assisted \ac{isac} systems and \ac{isac} systems without RIS, even when the number of active elements in the RIS is as few as $10\%$ of the total number of elements, without significantly increasing the total power consumption of the ISAC system.

	\section{System model} \label{sec:system}
	
	Consider a hybrid RIS assisted ISAC system with~\ac{dfbs} sensing $T$ targets while simultaneously serving $K$ users by transmitting both radar and communication signals, but with separate beamformers. We consider a narrowband scenario and model the  \ac{dfbs}  as a~\ac{ula} with elements separated by half of the signal wavelength. We begin by presenting the downlink transmit signal model, followed by modeling the hybrid RIS.

	\begin{figure}
		\centering
		\includegraphics[width=8cm,height=6cm]{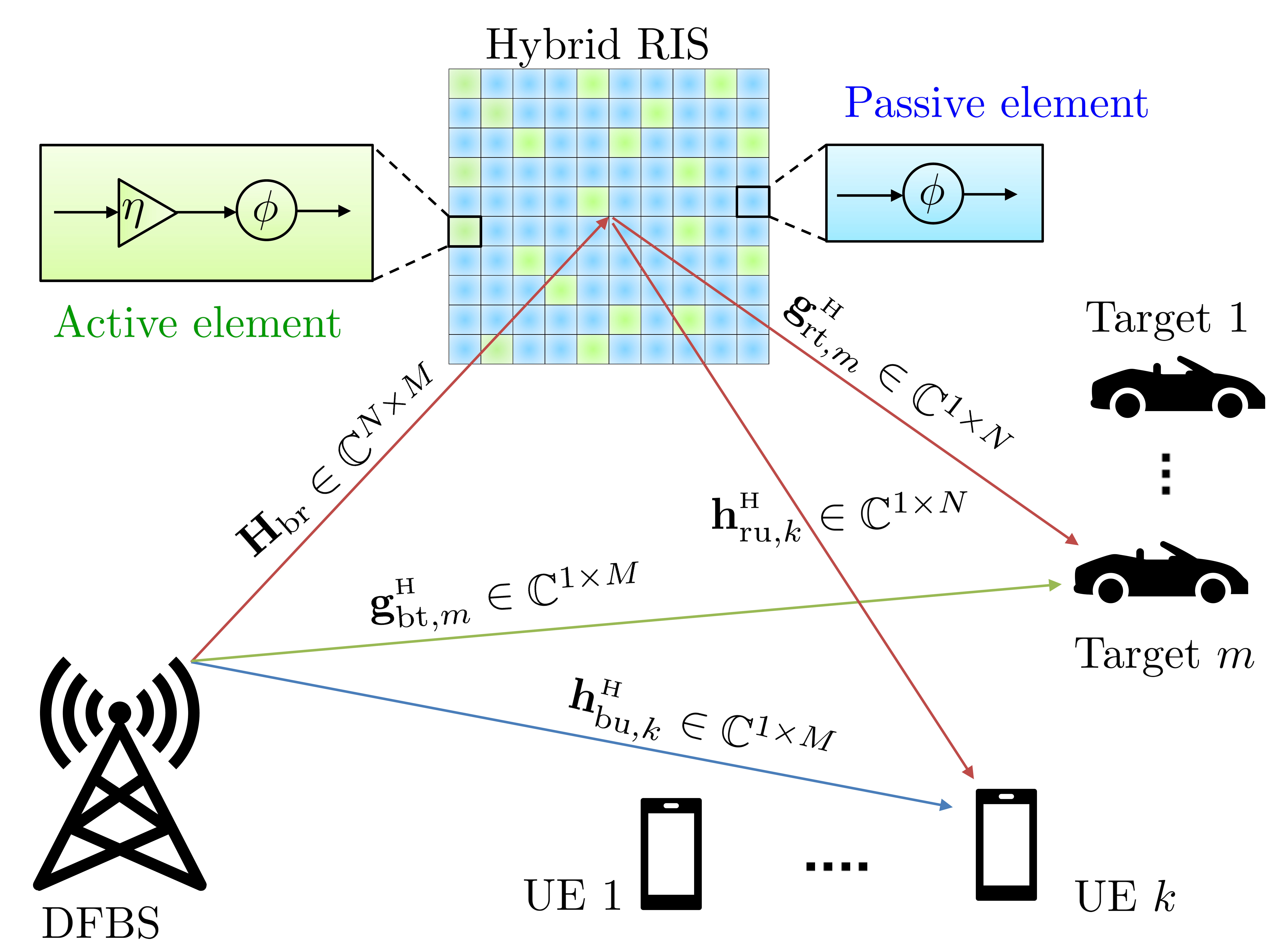}
		\caption{System model.} 
		\vspace{-0.5cm}
		\label{fig:sys_model}
	\end{figure}
	
	\subsection{Downlink transmit signal}
	Let $\vd = [d_1,\ldots,d_K]\rT \in \mbC^{K}$ and $\vt =[t_1,\ldots,t_M]\rT \in \mbC^{M}$ denote the discrete-time complex baseband signals used for communication and radar sensing, respectively. We further assume that these signals are uncorrelated with each other and have unit power, i.e., $\mbE \left[ \vd \vd \rH \right] = \mI_K$, $\mbE\left[ \vt \vt\rH \right] = \mI_M$, and $\mbE \left[ d_i t_j^* \right] = 0,~\forall~i,j$, where the expectation, $\mbE \left[\cdot\right]$, is  computed over different signal realizations. The \ac{dfbs} precodes $\vd$ using the communication beamformer $\mC = [\vc_1 ,\ldots ,\vc_K] \in \mbC^{M \times K}$ and precodes $\vt$ using the sensing beamformer $\mS =[\vs_1,\ldots,\vs_M] \in \mbC^{M \times M}$. Then the overall downlink transmit signal, $\vx$, is a superposition of the communication and radar symbols and is given by
	\begin{equation} \label{eq:tx_signal}
		\vx = \mC \vd + \mS \vt.
	\end{equation}
	The  corresponding transmit covariance matrix is denoted by $\mR = \mbE\left[ \vx \vx\rH\right]$.
	This signal reaches the users and targets via a direct channel and a reflected channel through the hybrid RIS, which we model next.
	\subsection{Hybrid RIS}
	Consider an $N$ element square-shaped hybrid RIS.
	Let $\boldsymbol{\omega} = [\omega_1,\ldots,\omega_N]\rT \in \mbC^{N}$ denote the RIS coefficients.  Without loss of generality, we assume that the first $L$ elements of the RIS are active with a maximum amplifier gain of $\eta$.  Let us denote the index set of active elements and passive elements as $\cA = \{1,\ldots,L\}$ and $\cP = \{L+1,\ldots,N\}$, respectively. Thus, we have $\vert \omega_i \vert \leq \eta$ for $i \in \cA$ and $\vert \omega_i \vert = 1$ for $i \in \cP$. Let $\mJ_{\rm a} \in \mbC^{N \times N}$ (respectively, $\mJ_{\rm p} \in \mbC^{N \times N}$ ) denote the selection matrix with first $L$ rows (respectively, last $N-L$ rows) being the same as that of the $N \times N$ identity matrix and zeros elsewhere. The output of the hybrid RIS for an input $\vu_{\rm in} \in \mbC^N$ is given by~\cite{zhang2022active,dong2021active_ris_secure} 
	\begin{equation} \label{eq:ris_out}
		\vu_{\rm out} = {\rm diag}(\boldsymbol{\omega}) \vu_{\rm in} + \mJ_{\rm a} {\rm diag}(\boldsymbol{\omega})\vn,
	\end{equation} 
	where $\vn \sim \cC\cN(\boldsymbol{0}, \nu^2 \mI)$ is the RIS noise due to the active elements. The total output power of the active elements is  $P_{\rm ris} =  \mbE \left[ \Vert \mJ_{\rm a} {\rm diag}(\boldsymbol{\omega}) \left( \vu_{\rm in} + \vn \right)  \Vert^2 \right]  $. Here, we re-emphasize that only the active elements result in RIS noise. 
	\subsection{Channel model}
	The communication channels involving the hybrid RIS are modeled as Rician fading channels with a Rician factor $\rho$, whereas the direct channels between the DFBS and the users are modeled as Rayleigh fading channels. All the other channels involving targets are modeled as~\ac{LoS} links. The system model with different underlying channels  is presented in Fig.~\ref{fig:sys_model}.
	
	\section{Problem statement}
	The communication performance of a multi-user MIMO system in terms of data rate or spectral efficiency is determined by the~\ac{sinr} at the user. Similarly, the sensing performance of a radar system is determined by the \ac{snr} at the radar receiver, which in turn predominantly depends on the target illumination power. Furthermore, the use of active elements in the RIS results in additional noise at the output of the RIS that is transmitted towards the  targets and users, thereby adversely affecting the radar and communication \ac{sinr}. In this work, we propose to choose the beamformers and RIS coefficients to maximize the worst-case target illumination power while guaranteeing a minimum \ac{sinr} for the communication users and constraining the RIS noise.
	
	Let $y_k = \vh_k\rH \vx +  \vh_{{\rm ru},k}\rH \mJ_{\rm a}{\rm diag}(\boldsymbol{\omega}) \vn + w_k$ denote the signal received at the $k$th UE, where $\vh_k = \vh_{ {\rm bu},k} + \vh_{{\rm ru},k} {\rm diag}(\boldsymbol{\omega})\mH_{\rm br}$ (definitions of channel matrices are provides in Fig.~\ref{fig:sys_model}) and $w_k \sim \cC\cN(0,\sigma^2)$ is the AWGN at the receiver. Using~\eqref{eq:tx_signal}, SINR of the $k$th user can be computed as
	\begin{equation*} 
		\gamma_k = \frac{  \vert \vh_k\rH \vc_k \vert^2  }{ \sum_{j=1, j\neq k}^{K} \vert \vh_k \rH \vc_j \vert^2 + \sum_{m=1}^{M} \vert \vh_k\rH \vs_m \vert^2 + z_k(\boldsymbol{\omega}) + \sigma^2 },
	\end{equation*}
	where  $z_k(\boldsymbol{\omega})$ is the RIS noise power at the $k$th UE, given by $z_k(\boldsymbol{\omega})  =    \mbE \left[ \vert \vh_{{\rm ru},k}\rH \mJ_{\rm a}{\rm diag}(\boldsymbol{\omega}) \vn  \vert^2\right]  =  \Vert \vh_{ {\rm ru},k} \rH \mJ_{\rm a} {\rm diag}(\boldsymbol{\omega}) \Vert^2 \nu^2$.

	For the $m$th target, the target illumination power arising from the signal transmitted by the DFBS is given by
	\begin{equation} \label{eq:target_spow}
		p_m(\mR,\boldsymbol{\omega}) = \mbE \left[\vert \vg_m \rH \vx \vert^2 \right] = \vg_m\rH \mR\vg_m = {\rm Tr}\left( \mR \vg_m \vg_m\rH \right),
	\end{equation}
	where $\vg_m\rH = \vg_{{\rm bt}, m}\rH + \vg_{{\rm rt}, m}\rH {\rm diag}(\boldsymbol{\omega})\mH_{\rm br}$ is the overall \ac{dfbs}-\ac{ris}-target channel vector. The RIS noise at the $m$th target is given by 
	\begin{equation} \label{eq:target_npow}
		r_m(\boldsymbol{\omega}) =   \Vert \vg_{ {\rm rt},m} \rH \mJ_{\rm a} {\rm diag}(\boldsymbol{\omega}) \Vert^2 \nu^2.
	\end{equation} 
	We design the transmit beamformers $\mC$ and $\mS$ (alternatively, the transmit covariance matrix $\mR$, which is positive semidefinite) and the coefficients of the hybrid RIS $\boldsymbol{\omega}$ while satisfying a total transmit power constraint of $P_t$ at the \ac{dfbs} and $P_{\rm max}$ at the RIS. Let $r_{\rm max}$ be the RIS noise power constraint for all targets and $\Gamma$ be the minimum SINR required for the UEs.  We can then state the  design problem as 
	\begin{subequations}  \label{problem_formulation}
		\begin{align}
			\quad \underset{\mS,\mC,\boldsymbol{\omega}}{\text{max}} & \quad \underset{m}{\text{min}} \quad p_m(\mR,\boldsymbol{\omega})  \nonumber \\ 
			\text{subject to} & \quad \mR = \mC\mC\rH + \mS\mS\rH \succcurlyeq \boldsymbol{0}  \nonumber \\ 
			& \quad \gamma_k(\mR,\boldsymbol{\omega}) \geq \Gamma, \quad k=1,2,\ldots,K  \label{problem_formulation_sinr}\\
			& \quad {\rm Tr}(\mR) \leq P_t  \label{problem_formulation_pow}\\
			& \quad r_m(\boldsymbol{\omega}) \leq r_{\rm max}, \quad m=1,\ldots, T \label{problem_formulation_tgt_ris_pow} \\		
			& \quad \vert \omega_{i} \vert \leq \eta,  i \in \cA, \quad \vert \omega_{i} \vert = 1, i \in \cP \label{problem_ris_constraint} \\
			& \quad P_{\rm ris}(\mR,\boldsymbol{\omega}) \leq P_{\rm max} \label{problem_ris_pow},
		\end{align}
	\end{subequations}
	where~\eqref{problem_formulation_sinr} is the SINR constraint for the users,~\eqref{problem_formulation_pow} is due to the total transmit power constraint at the \ac{dfbs},~\eqref{problem_formulation_tgt_ris_pow} ensures that the RIS noise power at each target is bounded, and~\eqref{problem_ris_pow} is the total power constraint at the RIS. Here,~\eqref{problem_ris_pow} also determines the required capacity of the power source at the RIS.
	Due to the total transmit power constraint at the \ac{dfbs} and limited gain of the active elements, the total power at the output of the RIS is upper bounded.  Specifically, the upper bound is provided in the next theorem. 
		\begin{theorem}
			The total power constraint on the active elements of the RIS is upper bounded as
			\begin{equation} \label{eq:lemma1}
				P_{\rm ris} \leq	L\eta^2 \left[ \zeta_{\rm br} P_t \left( \frac{\rho M + 1 }{\rho + 1}\right) + \nu^2 \right],
			\end{equation}
			where $\zeta_{\rm br}$ is the pathloss of the DFBS-RIS link.
		\end{theorem}
			\begin{proof} 
			 The output power at the $i$th active RIS element for the incident signal $\vu_{\rm in} = \mH_{\rm br}\vx$ is given by $q_i \overset{\Delta}{=}\mbE_{\vx, \vh_{{\rm br},i}}\left[ \vert \vh_{{\rm br},i}\rH \vx \vert^2 \right] + \nu^2$, where $\vh_{{\rm br},i}\rH$ denotes the $i$th row of $\mH_{\rm br}$. Let $\va(\theta_{\rm br}) \in \mbC^{M}$ be the array response vector at the DFBS towards the direction of RIS, $\theta_{\rm br}$, with $\Vert \va\left(\theta_{\rm br} \right)\Vert^2 = M$. From the assumed channel statistics, we have
				\begin{equation}
				q_i = {\zeta_{\rm br}} \left\{  \mbE_{\vx} \left[ \frac{\rho}{\rho + 1} \vert \va\rH(\theta_{\rm br})\vx \vert^2  + \frac{1}{\rho + 1} {\rm Tr}\left(\vx \vx\rH \right) \right]  \right\} + \nu^2, \nonumber 
				\end{equation}
			where the first term corresponds to the \ac{LoS} path and the second term corresponds to the \ac{NLoS} path. We have $\mbE_{\vx} \left[\vert  \va\rH( \theta_{\rm br}) \vx \vert^2 \right] = \va\rH( \theta_{\rm br})\mR \va \left(\theta_{\rm br} \right) \leq  P_tM$ leading to the upper bound $q_i \leq {\zeta_{\rm br}}P_t \left( \frac{ \rho M +1 }{\rho +1}\right) + \nu^2$. The total output power of the active RIS elements is $P_{\rm ris} = \sum_{i=1}^{L} \vert \omega_{i} \vert^2 q_i $. By
			replacing $q_i$ with its upper bound and from~\eqref{problem_ris_constraint}, we obtain~\eqref{eq:lemma1}.
			\end{proof}
Hence, if we choose $P_{\rm max}$ to be larger than the upper bound,~\eqref{problem_ris_pow} will be redundant and can be ignored. From now on, we ignore~\eqref{problem_ris_pow} in the design. We now proceed to the design of the transmit beamformers and RIS coefficients.

\section{Proposed solver} \label{sec:design}
	The optimization problem in~\eqref{problem_formulation} is non-linear and non-convex due to the quadratic inequality constraints and the interdependence between the optimization variables, rendering a joint solution difficult. Hence, we solve~\eqref{problem_formulation} in an alternating manner, where we optimize the beamformers for a fixed choice of the RIS coefficients and vice versa till convergence.
	
		\begin{figure*}[ht]
		\begin{subfigure}[c]{0.48\columnwidth}\centering
			\includegraphics[width=\columnwidth]{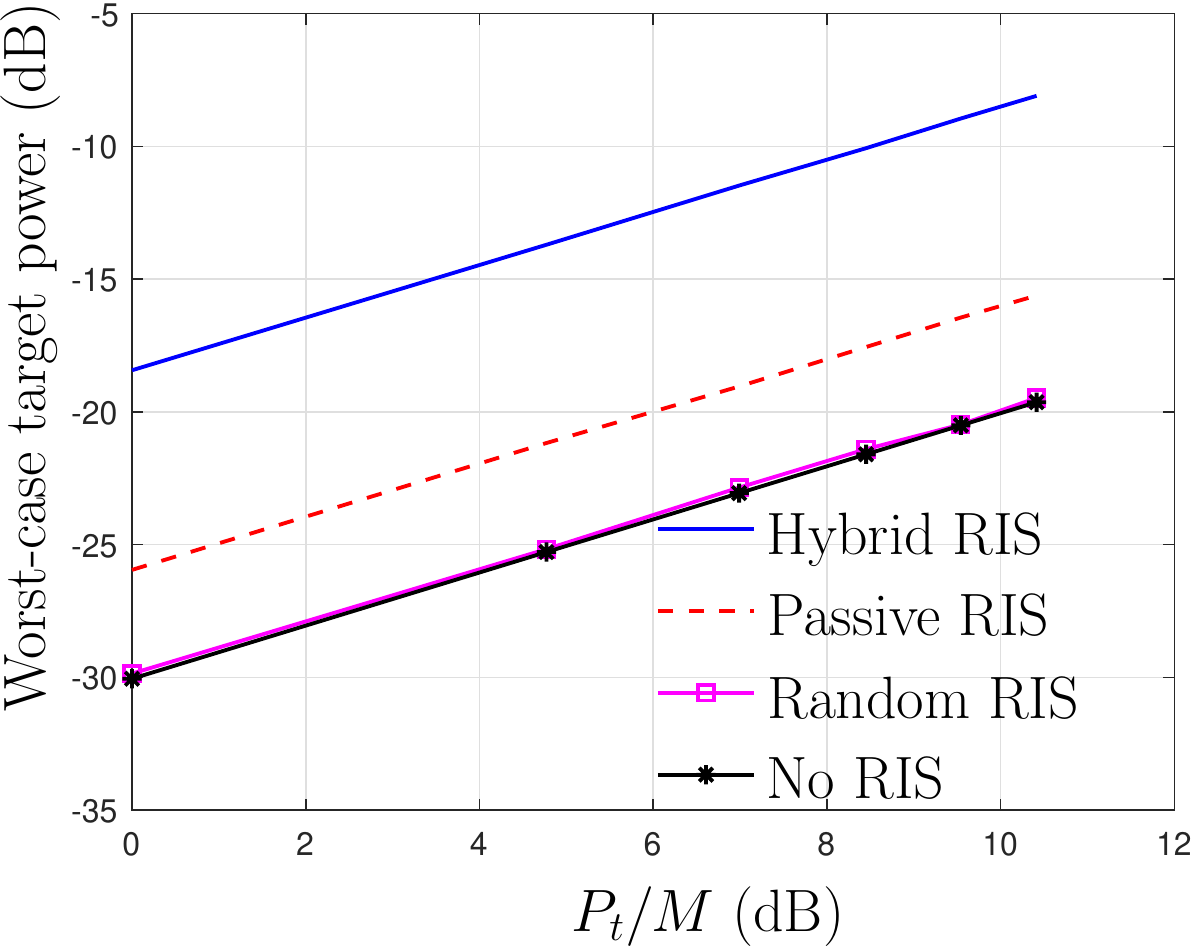}
			\caption{}
			\label{fig:f1:tgt_pow_v_Pt}
		\end{subfigure}
		~
		\begin{subfigure}[c]{0.48\columnwidth}\centering
			\includegraphics[width=\columnwidth]{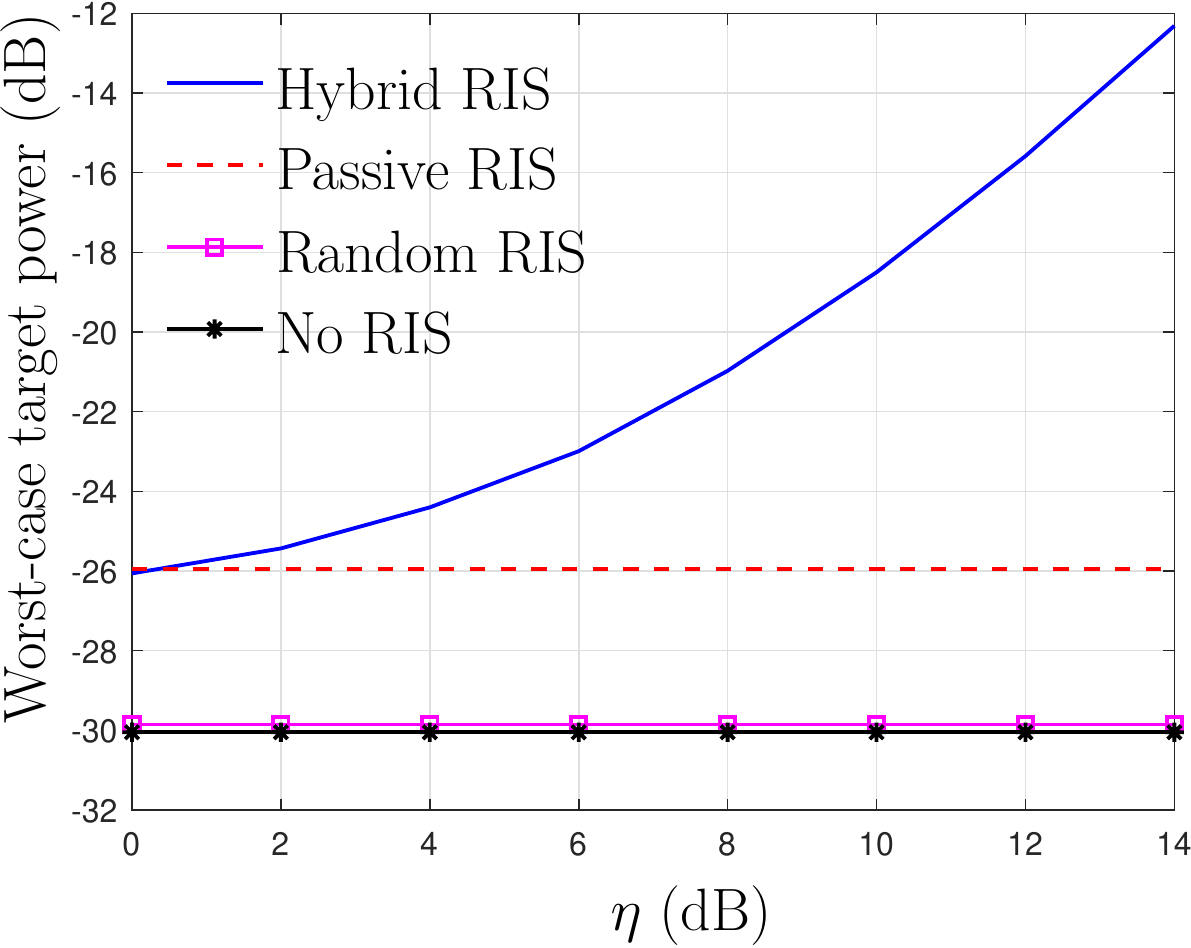}
			\caption{}
			\label{fig:f3:tgt_pow_v_eta}
		\end{subfigure}
		~
		\begin{subfigure}[c]{0.48\columnwidth}\centering
			\includegraphics[width=\columnwidth]{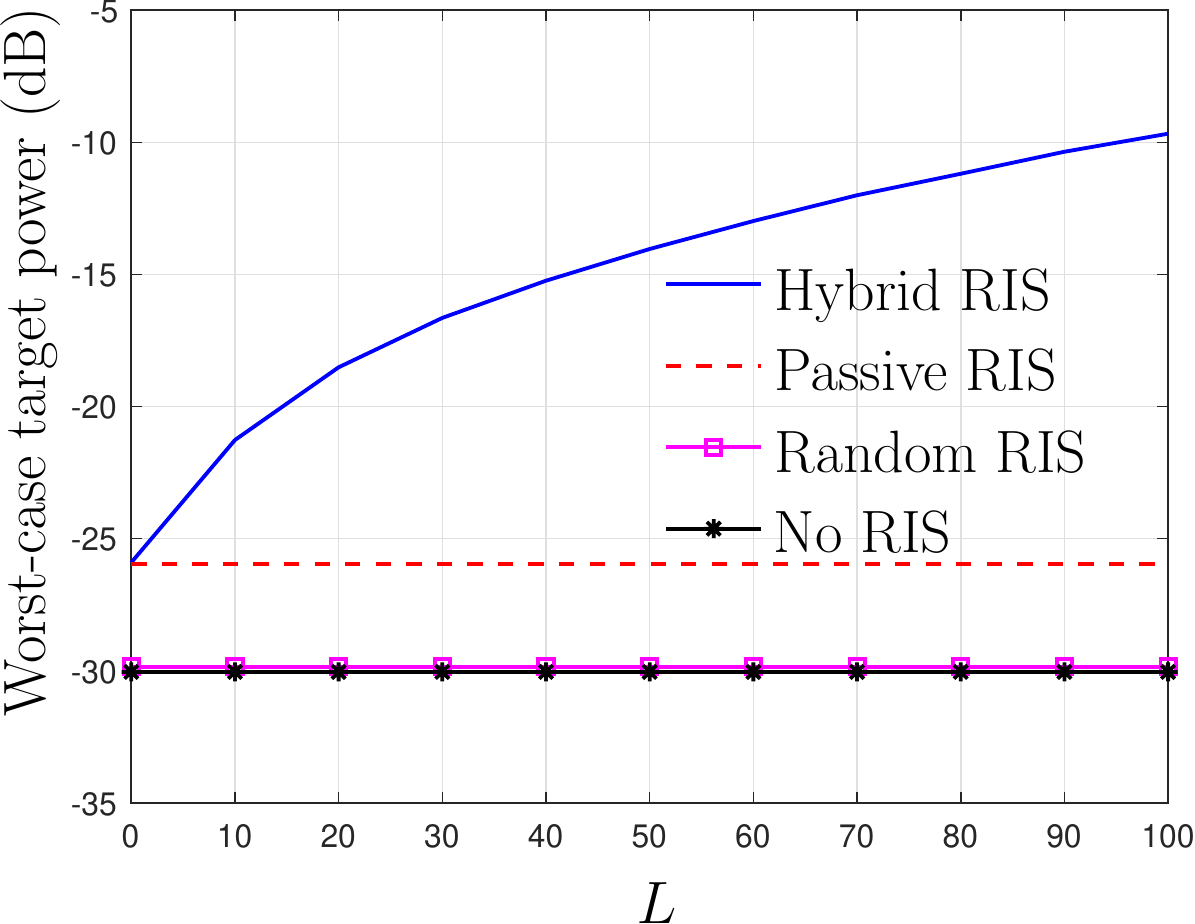}
			\caption{}
			\label{fig:f4:tgt_pow_v_L}
		\end{subfigure}
		~
		\begin{subfigure}[c]{0.48\columnwidth}\centering
			\includegraphics[width=\columnwidth]{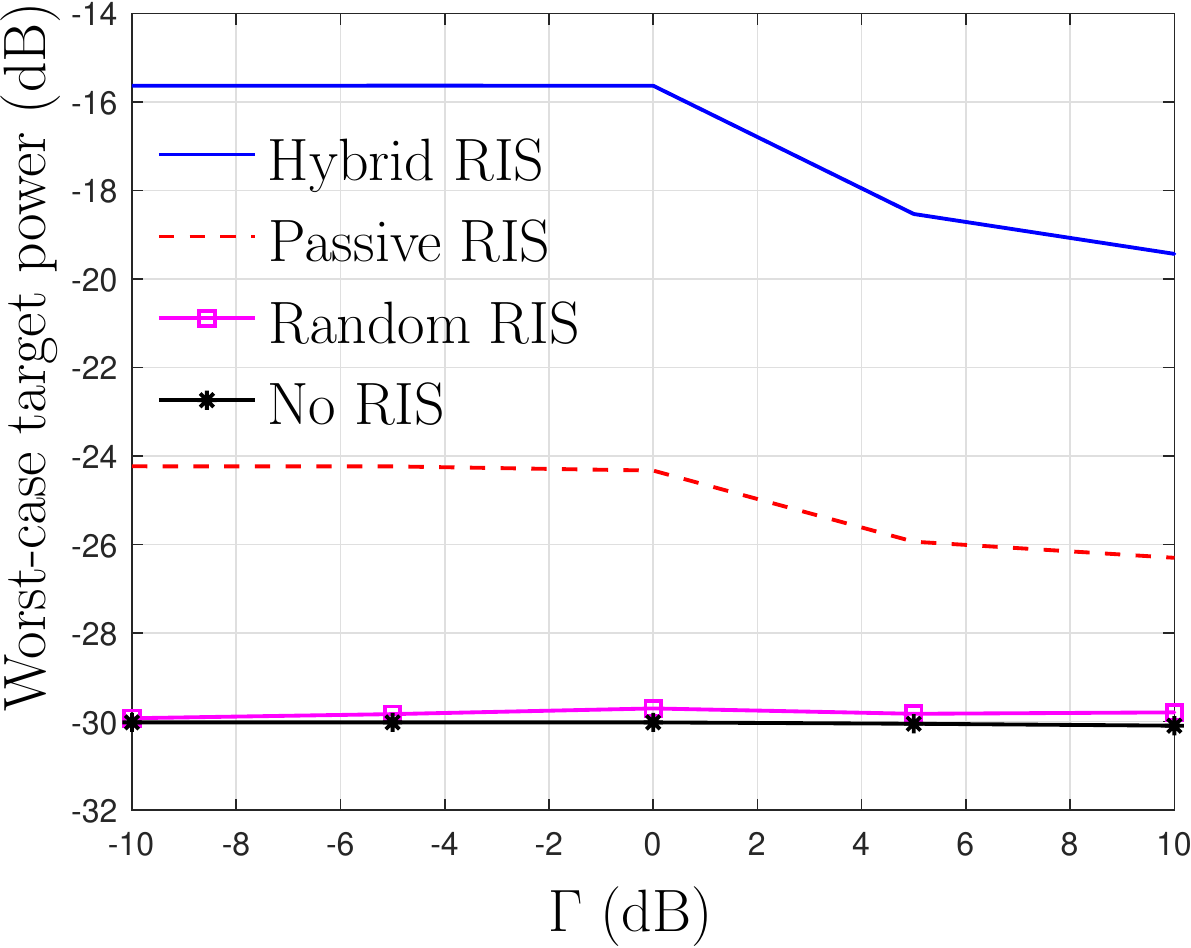}
			\caption{}
			\label{fig:f2:Ptotal_v_gamma}
		\end{subfigure}
		\caption{\small Worst-case target illumination power: (a) Impact of $P_t$ with $\eta = 10~{\rm dB}$ and $L = 20$, (b) Impact of $\eta$ with $P_t = M$ and $ L = 20$, (c) Impact of $L$ with $\eta = 10~{\rm dB}$ and $ P_t =  M$. (d) Impact of $\Gamma$ with $P_t = M, L = 20$, and $ \eta = 10$. We use $\Gamma = 5$~{\rm dB} and $r_{\rm max} = -90~{\rm dB}$.}
		\label{fig:simulations}
	\end{figure*}
	
	\vspace{-0.2cm}
	\begin{figure*} [hb]
		\rule{2\columnwidth}{0.1pt}
		\begin{equation} \label{eq:T_k_def}
			\vg_{{\rm c},mj}  = \begin{bmatrix}
				{\rm diag}\left( \vg_{ {\rm rt},m}\rH \right)\mH_{\rm br} \vc_j \\ \vg_{{\rm bt}, m}\rH \vc_j	\end{bmatrix}, \quad  \vg_{ {\rm r},mn } =  \begin{bmatrix}
				{\rm diag}\left( \vg_{ {\rm rt},m}\rH \right)\mH_{\rm br} \vs_n \\ \vg_{{\rm bt}, m}\rH \vs_n
			\end{bmatrix}, \quad 	\mT_m = \sum_{j=1}^{K} \vg_{{\rm c},mj} \vg_{{\rm c},mj}\rH + \sum_{n=1}^{M} \vg_{ {\rm r},mn } \vg_{ {\rm r},mn }\rH.
		\end{equation}

		\begin{equation} \label{eq:E_m_def}
			\vh_{{\rm c},kj} = \begin{bmatrix}
				{\rm diag}\left( \vh_{ {\rm ru},k}\rH \right)\mH_{\rm br} \vc_j \\ \vh_{{\rm bu}, k}\rH \vc_j
			\end{bmatrix}, 
			\quad  \vh_{ {\rm r},km } =  \begin{bmatrix}
				{\rm diag}\left( \vh_{ {\rm ru},k}\rH \right)\mH_{\rm br} \vs_m \\ \vh_{{\rm bu}, k}\rH \vs_m
			\end{bmatrix}, \quad  \mE_m = \nu^2 \begin{bmatrix}
				{\rm diag}\left( \vg_{ {\rm rt},m}\rH \right) 	{\rm diag}\left( \vg_{ {\rm rt},m} \right) & \boldsymbol{0} \\
				\boldsymbol{0}\rH & 0
			\end{bmatrix}.	
		\end{equation}	
		
		\begin{equation} \label{eq:B_k_def}
		\mA_k = \vh_{{\rm c},kk} \vh_{{\rm c},kk}\rH,  \quad
			\mB_k = \sum_{j=1,j\neq k}^{K} \vh_{{\rm c},kj} \vh_{{\rm c},kj}\rH + \sum_{m=1}^{M} \vh_{ {\rm r},km } \vh_{ {\rm r},km }\rH + \nu^2 \begin{bmatrix}
				{\rm diag}\left( \vh_{ {\rm ru},k} \rH \right) {\rm diag}\left( \vh_{ {\rm ru},k}  \right)  & \boldsymbol{0}\\
				\boldsymbol{0}\rH & 0
			\end{bmatrix}.
		\end{equation}

	\end{figure*}

\subsection{Updating $\mC$ and $\mS$ given $\boldsymbol{\omega}$} \label{sec:update_tx}

Let us define rank-1 matrices $\mC_k = \vc_k \vc_k \rH$ for $k=1,\ldots,K $ so that $\mC\mC\rH = \sum_{k=1}^{K}\mC_k$ and $\mS\mS\rH = \mR - \left( \sum_{k=1}^{K} \mC_k\right)$.{ \color{black} We can express the SINR constraints in~\eqref{problem_formulation_sinr} as~\cite{liu2020joint_transmit_beamform}
		\begin{equation}
			\left( 1 + \Gamma^{-1}\right) \vh_k \rH \mC_k \vh_k  \geq \vh_k \rH \mR \vh_k + z_k(\boldsymbol{\omega}) + \sigma^2 
	\end{equation}
	for $k=1,\ldots,K$.  } Then the beamformer design subproblem becomes

		\begin{align}
			\underset{\mR ,\mC_1,\ldots,\mC_K }{\text{max}}  \quad & \underset{m}{\text{min}} \quad {\rm Tr}\left( \mR \vg_m \vg_m \rH \right)   \nonumber \\ 
			\text{subject to} \quad  & \mC_k\succcurlyeq \boldsymbol{0}, \quad {\rm rank}\left(\mC_k \right) = 1   \label{sp1_rank_pow} \\
			& \left( 1 + \Gamma^{-1}\right) \vh_k \rH \mC_k \vh_k \geq \vh_k \rH \mR \vh_k  \nonumber\\
			&\quad \quad \quad \quad \quad \quad + z_k(\boldsymbol{\omega}) + \sigma^2, \quad k=1,\ldots,K  \nonumber \\
			& \quad \mR \succcurlyeq \boldsymbol{0}, \bigskip {\rm Tr}(\mR) \leq P_t, \bigskip \left( \mR  - \sum_{k=1}^{K} \mC_k \right)   \succcurlyeq \boldsymbol{0}.  \nonumber 		
		\end{align}
On dropping the non-convex rank constraints, the above problem reduces to an SDP, which can be efficiently solved using off-the-shelf solvers.	To recover rank-1 matrices $\mC_k$ from the SDP solution, we adopt the procedure from~\cite{liu2020joint_transmit_beamform} as described next. Suppose $\hat{\mR}$ and $\hat{\mC}_k,~k=1,\ldots,K$,  is the solution to the relaxed problem. We construct the beamformers by introducing the rank-1 matrices
	\begin{equation} \label{eq:ck_construct}
		\tilde{\vc}_k = \frac{\hat{\mC}_k\vh_k}{\sqrt{ \vh_k \rH \hat{\mC}_k \vh_k }}; \quad \tilde{\mC}_k = \tilde{\vc}_k \tilde{\vc}_k\rH,
	\end{equation}
	{\color{black} which satisfies  $\vh_k\rH\tilde{\mC}_k\vh_k = \vh_k\rH \hat{\mC}_k\vh_k$ $\text{for}\quad k=1,\ldots,K$.  Since $z_k(\boldsymbol{\omega})$ and $\sigma^2$ are known constants independent of the beamformers, \eqref{eq:ck_construct} also satisfies the SINR constraints.}
	We then use the Cholesky decomposition to obtain the sensing beamformers as
	\begin{equation} \label{eq:s_construct}
		\left( \hat{\mR} - \sum_{k=1}^{K} \tilde{\mC}_k \right) = \tilde{\mS} \tilde{\mS}\rH.
	\end{equation}

	{\color{black} Using the Cauchy-Schwartz inequality, we can show that $\left( \hat{\mR}  - \sum_{k=1}^{K} \tilde{ \mC_k } \right)   \succcurlyeq \boldsymbol{0} $.  Thus, in other words, solution obtained using~\eqref{eq:ck_construct} and~\eqref{eq:s_construct} is a feasible solution to the considered optimization problem~\eqref{sp1_rank_pow} whose  objective function is solely determined by $\mR$. Hence, the value of objective function with the constructed solution will be same as that of the relaxed solution. Before concluding this section, we re-emphasize that the constraints~\eqref{problem_formulation_tgt_ris_pow} and \eqref{problem_ris_constraint} that solely depend on the RIS coefficients are not relevant for the beamformer design subproblem since the RIS coefficients are fixed. 	

 }

	\subsection{Updating $\boldsymbol{\omega}$  given $\mC$ and $\mS$} \label{sec:update_ris}
	We now update the hybrid RIS coefficients while fixing the beamformers. Unlike the design procedure involving passive RIS, the presence of the RIS noise terms and the total power constraint on the RIS elements complicate the design of RIS coefficients. While the total power constraint for the RIS~\eqref{problem_ris_pow} can be redundant, the RIS noise power terms need to be accounted for in the design of RIS coefficients.
	
	To begin with, we express the  received power at the targets as a quadratic function of RIS coefficients. To do that, let us define the variable $\vv\rH = [\boldsymbol{\omega}\rT , 1] \in \mbC^{1 \times (N+1)}$. Then, the signal power received at the $m$th target can be written as $p_m = \vv\rH \mT_m \vv = {\rm Tr}(\mT_m \vv \vv\rH)$, where $\mT_m$ is given by~\eqref{eq:T_k_def}. Similarly, we define matrices $\mE_m$ in~\eqref{eq:E_m_def}, $\mA_k$ and $\mB_k$ in~\eqref{eq:B_k_def}, and express the RIS coefficient design  subproblem as
	\begin{subequations} \label{sp2}
		\begin{align} 
			\quad	\underset{\mV = \vv\vv\rH }{\text{max}} \quad \underset{m}{\text{min}}& \quad {\rm Tr}\left( \mT_m \mV \right) \nonumber  \\ 
			{\rm Tr}\left(  \left(\mA_k - \Gamma \mB_k \right)\mV \right) &\geq \Gamma \sigma^2, \quad k=1,\ldots, K \label{sp2_sinr}\\
			{\rm Tr}\left( \mE_m \mV \right) &\leq r_{\rm max}, \quad m=1,\ldots,T  \label{sp2_tgtnoise}\\
			\quad \vert [\mV]_{i,i} \vert \leq \eta,  i \in \cA, & \quad  \vert [\mV]_{i,i} \vert = 1, i \in \cP. \label{sp2_ris}   
		\end{align} 
	\end{subequations}
	We can re-write the quadratic equality $\mV = \vv \vv\rH $ equivalently as $\mV \succcurlyeq \boldsymbol{0}$ along with ${\rm rank}\left(\mV \right)=1$. By dropping the rank constraint, we get an SDP, which can be solved using off-the-shelf solvers. The required rank-1 solution for the RIS coefficients can be obtained using Gaussian randomization~\cite{luo2010sdr}. Specifically, let $\mV_{\rm opt}$ be the solution to the relaxed version of~\eqref{sp2} and $\tilde{\mV}_{\rm opt}$ be the $N\times N$ top left submatrix of $\mV_{\rm opt}$. We draw multiple realizations of Gaussian random vectors $\vu_m \sim \cC\cN(\boldsymbol{0},\tilde{\mV}_{\rm opt})$ and normalize elementwise to satisfy~\eqref{sp2_ris}, and choose a subset of these realizations, say, $\{{\vu_m}\}_{m=1}^{N_{\rm rand}}$ satisfying~\eqref{sp2_sinr} and~\eqref{sp2_tgtnoise}. The RIS reflection coefficients are then obtained by selecting the realization that results in the maximum worst-case target illumination power.
	
	To summarize, we repeat the following steps till convergence: (a) fix RIS coefficients and update beamformers using~\eqref{eq:ck_construct} and~\eqref{eq:s_construct} and (b) update RIS coefficients by solving~\eqref{sp2} followed by Gaussian randomization. This completes the design of beamformers and combiners.
	
	\section{Numerical experiments} \label{sec:numerical_experiments}
	In this section, we present numerical simulations to demonstrate the performance of the proposed algorithm. Throughout the simulations, we use $M=16$, $N=100$, $\Gamma = 5$~dB, $K=2$, $T = 4$, $L=20$, and $\eta = 10~{\rm dB}$, unless otherwise mentioned. The \ac{dfbs} and \ac{ris} are  located at $(0,0,0)$~m and $(10,-8,5)$~m, respectively. The user and target locations are randomly generated from a $10\times10~{\rm m}^2$ rectangular area with bottom left corner located at $(5,-2,0)$~m. The pathloss for the direct links and links involving the RIS are modeled as $
	30 + 22 \log_{10} (d)~{\rm dB}$ and $
	30 + 35 \log_{10} (d)~{\rm dB}$, respectively, where $d$ is the distance between the terminals in m. The Rician factor $\rho$ is set to $10$. We set $\sigma^2 = -94~{\rm dBm}$ and $\nu^2 = -60~{\rm dBm}$ and use $10$ iterations of the proposed algorithm. All plots are obtained by averaging over 100 independent wireless channel realizations.
	
	We compare the performance of the proposed method with three benchmark schemes: (a) \texttt{passive RIS}, which is a special case of the proposed method with $L=0$, (b) \texttt{random RIS}, where we design the beamformers $\mS$ and $\mC$ for a random choice of passive RIS phase shifts $\boldsymbol{\omega}_c$, and (d) \texttt{No RIS}, which is an ISAC system without any RIS. We simulate \texttt{No RIS} scenario by setting RIS phase shifts to zero.  For comparison, the same cost function and constraints are used throughout the benchmark schemes, and we use the worst-case target illumination power as the (radar) performance metric.

	In Fig.~\ref{fig:simulations}(a), we present the performance of different systems by varying $P_t$. We can clearly observe that the hybrid-RIS-assisted ISAC system with an amplification factor of $\eta=10~{\rm dB}$ and just $L=20$ active elements significantly outperforms both passive RIS assisted as well as ISAC systems without RIS. Furthermore, the performance of \texttt{random RIS} is comparable to that of \texttt{No RIS} scenario, thereby demonstrating the need for an appropriate design to completely benefit from the RIS. The performance of all methods improves when we have a higher transmit power.
	
	Impact of amplification factor (i.e.,~$\eta$) and the number of active elements of the hybrid RIS~(i.e., $L$) on the performance is presented in Fig.~\ref{fig:simulations}(b) and Fig.~\ref{fig:simulations}(c), respectively. Performance of the hybrid RIS assisted ISAC system improves with an increase in $\eta$ as well as $L$ and can be attributed to two phenomena. Firstly, a  larger value of $\eta$ leads to an increase in the strength of paths due to increased amplification. Secondly, increasing $L$ leads to an increase in the degrees of freedom in the RIS coefficient design since more number of elements can now take values that are not constrained to be on the unit circle, leading to improved beamforming performance and possibly stronger amplification. As before, \texttt{hybrid RIS} is found to remarkably outperform all benchmark schemes.
	
	Impact of SINR~(i.e.,~$\Gamma$) on the radar performance is illustrated in Fig.~\ref{fig:simulations}(d). As $\Gamma$ increases, the radar performance decreases since a higher amount of power needs to be transmitted towards the users to ensure a desired SINR. This tradeoff is inherent in ISAC systems since the total power available for communication and sensing is limited. The aforementioned tradeoff is more prominent at higher SINRs due to increased signal power requirement at the users. As before, the hybrid RIS assisted ISAC system is significantly better than the benchmark schemes.
	
	To compute the additional power consumed by the hybrid RIS, we substitute the parameters used for the simulation in~\eqref{eq:lemma1}. For $P_t/M = 0~{\rm dB}$, a hybrid RIS with an amplifier gain of $\eta = 10~{\rm dB}$ and $L=20$ is consuming an additional power about $-12~{\rm dB}$ while improving the target illumination power by $11~{\rm dB}$.  However, from Fig.~\ref{fig:simulations}(a), a fully passive RIS requires $P_t/M =7~{\rm dB}$ to achieve similar performance, which is significantly high. To conclude, using a hybrid RIS leads to  significant improvements in the performance of ISAC systems without significantly increasing the power consumption.
	\section{Conclusions} \label{sec:conclusions}
	In this paper, we considered a hybrid RIS to enhance the performance of an ISAC system serving multiple users and targets. Specifically, we designed  transmit beamformers and RIS coefficients to maximize the worst-case target illumination power while ensuring a minimum~\ac{sinr} for the users and keeping the RIS noise power bounded. We used an alternating optimization based scheme and designed beamformers by fixing RIS coefficients and vice versa. Specifically, we relaxed each of the subproblems to obtain an SDP, which was then solved using off-the-shelf solvers. The worst-case target illumination power with the proposed hybrid RIS assisted ISAC systems is found to be remarkably higher than that of passive RIS assisted ISAC systems as well as ISAC systems without \ac{ris} even when only a small fraction of the hybrid \ac{ris} elements are active.

	\bibliographystyle{ieeetran}
	\bibliography{IEEEabrv,bibliography}

\end{document}